\documentclass{LMCS}

\def\doi{8(3:27)2012}
\lmcsheading%
{\doi}
{1--15}
{}
{}
{Dec.~10, 2010}
{Sep.~29, 2012}
{}
 
\usepackage{amsfonts}
\usepackage{tikz}
\usepackage{microtype}
\usepackage{enumerate,hyperref}

\usetikzlibrary{snakes}

\newcommand{\step}[1]{[#1\rangle}
\newcommand{\nat}{\mathbb{N}}
\newcommand{\Z}{\mathbb{Z}}
\newcommand{\then}{\Longrightarrow}
\newcommand{\parikh}{\wp}
\newcommand{\I}{C} 
\newcommand{\C}{{\ensuremath \Gamma}} 
\newcommand{\ord}{\ensuremath{\Omega}}

\def\hascolor#1#2{#2} 
\hascolor{
\tikzstyle{place}=[circle,draw=black!50,fill=black!10,thick,inner sep=0pt,minimum size=6mm]
\tikzstyle{oplace}=[circle,draw=green!50,fill=green!10,thick,inner sep=0pt,minimum size=6mm]
\tikzstyle{iplace}=[circle,draw=red!50,fill=red!10,thick,inner sep=0pt,minimum size=6mm]
\tikzstyle{transition}=[rectangle,draw=black!50,fill=black!20,thick,inner sep=0pt,minimum size=6mm]
\tikzstyle{token}=[circle,draw=black,fill=black,inner sep=0pt,minimum size=1mm]
\tikzstyle{arrow}=[-latex]
\tikzstyle{myred}=[red]
\tikzstyle{mygreen}=[green]
}{
\usetikzlibrary{patterns}
\tikzstyle{place}=[circle,draw=black,thick,inner sep=0pt,minimum size=6mm]
\tikzstyle{gplace}=[circle,draw=black,fill=black!40,thick,inner sep=0pt,minimum size=6mm]
\tikzstyle{oplace}=[circle,draw=black,pattern=north west lines,pattern color=black!30,thick,inner sep=0pt,minimum size=6mm]
\tikzstyle{iplace}=[circle,draw=black,pattern=north east lines,pattern color=black!30,thick,inner sep=0pt,minimum size=6mm]
\tikzstyle{transition}=[rectangle,draw=black,thick,inner sep=0pt,minimum size=6mm]
\tikzstyle{gtransition}=[rectangle,draw=black,fill=black!40,thick,inner sep=0pt,minimum size=6mm]
\tikzstyle{otransition}=[rectangle,draw=black,pattern=north west lines,pattern color=black!30,thick,inner sep=0pt,minimum size=6mm]
\tikzstyle{itransition}=[rectangle,draw=black,pattern=north east lines,pattern color=black!30,thick,inner sep=0pt,minimum size=6mm]
\tikzstyle{token}=[circle,draw=black,fill=black,inner sep=0pt,minimum size=1mm]
\tikzstyle{arrow}=[-latex]
\tikzstyle{myred}=[black!60]
\tikzstyle{mygreen}=[black!40]
}

\begin{document}

\title[Applying CEGAR to the Petri Net State Equation]{Applying CEGAR to the Petri Net State Equation}

\author[H.~Wimmel]{Harro Wimmel} 
\author[K.~Wolf]{Karsten Wolf}

\address{Universit\"at Rostock, Institut f\"ur Informatik}
\email{\{harro.wimmel, karsten.wolf\}@uni-rostock.de}

\begin{abstract}
We propose a reachability verification technique that combines the {\em 
Petri net state equation} (a linear algebraic
overapproximation of the set of reachable states) with the concept of
{\em counterexample guided abstraction refinement}. In essence, we 
replace the search through the set of reachable
states by a search through the space of solutions of the state equation. 
We demonstrate the excellent performance of the
technique on several real-world examples. The technique is particularly 
useful in those cases where the reachability query
yields a negative result: While state space based techniques need to 
fully expand the state space in this case, our
technique often terminates promptly. In addition, we can derive some 
diagnostic information in case of unreachability
while state space methods can only provide witness paths in the case of 
reachability.
\end{abstract}

\keywords{
Petri Net, Reachability Problem, Integer Programming, CEGAR, Structure
Analysis, Partial Order Reduction.}
\subjclass{F.2.2, I.6.4}

\maketitle
\enlargethispage*{2\baselineskip}

\section{Introduction}

Reachability is {\em the} fundamental verification problem.
For place/transition Petri nets (which may have infinitely many states), it is one of the hardest decision problems 
known among the naturally emerging yet decidable problems
in computer science. General solutions have been found by
Mayr~\cite{mayr84} and Kosaraju~\cite{kosaraju82} with later simplifications made by Lambert~\cite{lambert92}, but there are
complexity issues. All these approaches use coverability graphs which can have a non-primitive-recursive
size with respect to the corresponding Petri net. A new approach by Leroux~\cite{leroux09} not using
such graphs gives some hope, but a concrete upper bound for the worst case complexity so far
eludes us. In a sense even worse, Lipton~\cite{lipton76} has shown that the problem is {\sf EXPSPACE}-hard,
so any try at programming a tool efficiently solving this problem to the full extent must surely fail.

Nevertheless, efficient tools exist that are applicable to a considerable number of problem instances.
Model checkers, symbolic~\cite{cms06} or with partial order reduction~\cite{Wolf_2007_icatpn}, have been used
successfully to solve quite large reachability problems. On a positive answer, a model checker can typically generate
a trace, i.e. a firing sequence leading to the final marking. In contrast, negative answers are usually
not accompanied by any diagnostic information. Such information, i.e.\ a counterexample or reasoning
why the problem has a negative solution would require a deep analysis of the structure of the
Petri net. So far, no tools are known that analyze the structure of a net and allow for such
reasoning. 

This paper presents an approach to the reachability problem that combines two existing methods. First,
we employ the {\em state equation} for Petri nets. This is a linear-algebraic overapproximation on the
set of reachable states. Second, we use the concept of {\em counterexample guided abstraction refinement}
(CEGAR) \cite{cegar} for enhancing the expressiveness of the state equation. In essence, we iteratively analyse spurious
solutions of the state equation and add constraints that exclude a solution found to be spurious but do not exclude
any real solution. The approach has several advantages compared to (explicit or symbolic) purely state space based
verification techniques:
\begin{iteMize}{$\bullet$}
\item The search is quite focussed from the beginning as we traverse the solution space of the state equation rather than the set of reachable states;
\item The search is close to breadth-first traversal, so small witness traces are generated;
\item The method may perform well on unreachable problem instances (where state space techniques compute maximum size state spaces);
\item In several unreachable problem instances, some kind of diagnostic information can be provided; 
\item A considerable workload can be shifted to very mature tools for solving linear programming problems.
\end{iteMize}

\noindent In Sect.~\ref{sec2} we give the basic
definitions. Section~\ref{sec3} shows how to use integer programming tools to find candidates
for a solution. Section~\ref{sec4} deals with the analysis of the Petri net structure that
is needed to push the integer programming onto the right path. In Sect.~\ref{sec5} we
use methods of partial order reduction to mold the results of the integer programming
into firing sequences solving the reachability problem. 
In Sect.~\ref{sec6} the overall algorithm is presented and in Sect.~\ref{sec7} we drop
a few hints on how and when diagnostic information for unreachability can be generated.
Finally, Sect.~\ref{sec8} compares the results of an implementation
with another model checker, showing that structure analysis can compete with other approaches.

\section{The Reachability Problem}\label{sec2}

\begin{defi}[Petri net, marking, firing sequence]
A {\em Petri net} $N$ is a tuple $(S,T,F)$ with a set $S$ of {\em places}, a set $T$ of {\em transitions},
where $S\neq\emptyset\neq T$ and $S\cap T=\emptyset$, and a mapping $F$: $(S\times T)\cup(T\times S)\to\nat$
defining {\em arcs} between places and transitions.

A {\em marking} or {\em state} of a Petri net is a map $m$: $S\to\nat$. A place $s$ is said to contain $k$
{\em tokens} under $m$ if $m(s)=k$. A transition $t\in T$ is 
{\em enabled under $m$}, $m\step{t}$, if $m(s)\ge F(s,t)$ for every $s\in S$. A transition $t$
{\em fires under $m$ and leads to $m'$}, $m\step{t}m'$, if additionally $m'(s)=m(s)-F(s,t)+F(t,s)$
for every $s\in S$.

A word $\sigma\in T^*$ is a {\em firing sequence under $m$ and leads to $m'$}, $m\step{\sigma}m'$,
if either $m=m'$ and $\sigma=\varepsilon$, the empty word, or $\sigma=wt$, $w\in T^*$, $t\in T$
and $\exists m''$: $m\step{w}m''\step{t}m'$. A firing sequence $\sigma$ under $m$ is enabled under $m$,
i.e.\ $m\step{\sigma}$. The {\em Parikh image} of a word $\sigma\in T^*$ is the vector
$\parikh(\sigma)$: $T\to\nat$ with $\parikh(\sigma)(t)=\#_t(\sigma)$, where $\#_t(\sigma)$ is the number of
occurrences of $t$ in $\sigma$. For any firing sequence $\sigma$, we call $\parikh(\sigma)$ {\em realizable}.
\end{defi}

As usual, places are drawn as circles (with tokens as black dots inside them), transitions as rectangles, 
and arcs as arrows with $F(x,y)>0$ yielding an arrow pointing from $x$ to $y$. If an arc has a weight of
more than one, i.e. $F(x,y)>1$, the number $F(x,y)$ is written next to the arc. In case $F(x,y)=F(y,x)>0$,
we may sometimes draw a line with arrowheads at both ends.

Note, that the Parikh image is not an injective function. Therefore, $\parikh(\sigma)$ can be realizable
even if $\sigma$ is not a firing sequence (provided there is another firing sequence $\sigma'$
with $\parikh(\sigma)=\parikh(\sigma'))$.

\begin{defi}[Reachability problem]
A marking $m'$ is {\em reachable} from a marking $m$ in a net $N=(S,T,F)$ if there is a firing sequence 
$\sigma\in T^*$ with $m\step{\sigma}m'$. A tuple $(N,m,m')$ of a net and two markings is called a
{\em reachability problem} and has the answer ``yes'' if and only if $m'$ is reachable from $m$ in $N$.
The set ${\sf RP} = \{(N,m,m')\,|\,N$ is a Petri net, $m'$ is reachable from $m$ in $N\}$ is generally 
called {\em the} reachability problem, for which membership is to be decided.
\end{defi}

It is well-known that a necessary condition for a positive
answer to the reachability problem is the feasibility of the {\em state equation}.

\begin{defi}[State equation]
For a Petri net $N=(S,T,F)$ let $\I\in\nat^{S\times T}$, defined by $\I_{s,t}=F(t,s)-F(s,t)$, be the
{\em incidence matrix} of $N$. For two markings $m$ and $m'$, the system of linear equations
$m+\I x=m'$ is the {\em state equation} of $N$ for $m$ and $m'$. A vector $x\in \nat^T$ fulfilling
the equation is called a {\em solution}.
\end{defi}

\begin{prop}
For any firing sequence $\sigma$ of a net $N=(S,T,F)$ leading from $m$ to $m'$, i.e. $m\step{\sigma}m'$,
holds $m+\I\parikh(\sigma)=m'$, i.e.\ the Parikh vector of $\sigma$ is a solution of the state equation for
$N$, $m$, and $m'$. 
\end{prop}
This is just a reformulation of the firing condition for $\sigma$.

\begin{prop}
If the Petri net is acyclic, i.e.~the transitive closure of $F$ is irreflexive then the existence of a
solution of the state equation for $N$, $m$, and $m'$ is a sufficient condition for reachability of $m'$
from $m$ in $N$.
\end{prop}
In an acyclic net, minimal transitions (with respect to $F$) that occur in the support of a solution $x$ of the state equation
must be enabled under $m$. Firing such a transition with resulting marking $m^*$ leads to a smaller
reachability problem $(N,m^*,m')$ that has $x$ minus one occurrence of $t$ as a solution. By induction, the whole
solution can be unwound to a firing sequence.

In Petri nets with cycles, it is possible to have a sequence $\sigma$ such that its Parikh image fulfills the state equation but it is
not a firing sequence. 
The easiest example for this occurs in a net $N=(\{s\},\{t\},F)$ with $F(s,t)=1=F(t,s)$.
Let $m$ and $m'$ be the {\em empty marking}, i.e. one with zero tokens overall, then $m\step{t}m'$ is
obviously wrong but $m+\I\parikh(\sigma)=m'$ holds since $\I=(0)$. The effect can occur whenever the Petri net
contains a cycle of transitions. Interestingly, certain cycles of transitions can also help to overcome
this problem, see Fig.~\ref{f.tinv}. Here, we would like to fire a word $tt'$ from the marking $m$ with
$m(s_1)=m(s_2)=0$ and $m(s_3)=1$, but obviously, this is impossible. If we interleave, however, $tt'$ with the sequence $uu'$
we can fire $utt'u'$. The sequence $utt'u'$ corresponds to a solution of the state equation for $N$, $m$, and $m'$ that is not minimal.
More precisely, we have $\I\parikh(uu') = 0$. At first glance, such a sequence does not change the marking and appears to be neglegible.
However, it has a valuable effect on $tt'$ in ``lending a token'' to $s_2$. The transition $u$ provides that token, $u'$ takes it back, but meanwhile the token
helps the sequence $tt'$ to proceed. The process of ``lending tokens'' is not visible in the state equation as the latter overapproximates the
token game of Petri nets to linear algebra. In consequence, it is necessary for our approach to consider non-minimal solutions of the state equation
and in particular solutions to the corresponding homogeneous system of equations.

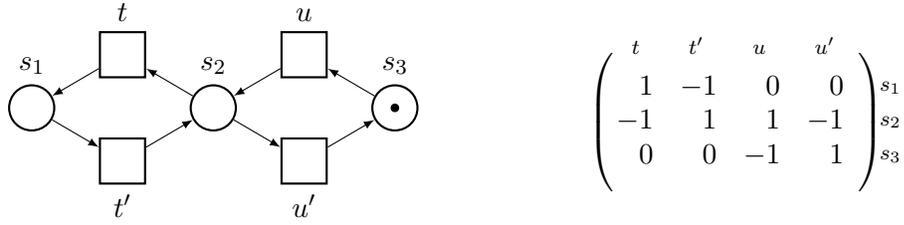
\begin{figure}[tb]
\centering
\begin{tikzpicture}[scale=0.5]
\node[transition,label=above:$t$] (t) {};
\node[place,label=above:$s_2$] (s2) [below right of=t,xshift=5mm] {};
\node[transition,label=below:$t'$] (tx) [below left of=s2,xshift=-5mm] {};
\node[place,label=above:$s_1$] (s1) [below left of=t,xshift=-5mm] {};
\node[transition,label=below:$u'$] (ux) [below right of=s2,xshift=5mm] {};
\node[transition,label=above:$u$] (u) [above right of=s2,xshift=5mm] {};
\node[place,label=above:$s_3$] (s3) [below right of=u,xshift=5mm] {};
\draw (s3) node[token] {};
\draw[arrow] (s2) to (t);
\draw[arrow] (t) to (s1);
\draw[arrow] (s1) to (tx);
\draw[arrow] (tx) to (s2);
\draw[arrow] (s2) to (ux);
\draw[arrow] (ux) to (s3);
\draw[arrow] (s3) to (u);
\draw[arrow] (u) to (s2);
\end{tikzpicture}
\hspace*{2cm}
\unitlength1cm
\begin{picture}(3,1.5)(0,-1.3)
\put(0,0){$\left(\begin{array}{c}\hspace*{2.9cm}\\~\\~\\~\end{array}\right)$}
\put(0.2,0){$\begin{array}{rrrr}1&-1&0&0\\ -1&1&1&-1\\ 0&0&-1&1\end{array}$}
\put(3.7,0.05){$\begin{array}{r}\scriptstyle s_1\\ \scriptstyle s_2\\ \scriptstyle s_3\end{array}$}
\put(0.4,1.0){$\begin{array}{rrrr}\scriptstyle t\;\;\;\,&\scriptstyle t'\;\;\;\,&\scriptstyle u\;\;\;\,&\scriptstyle u'\end{array}$}
\end{picture}
\caption{\label{f.tinv}The word $tt'$ cannot fire, but we can borrow a token from the circle $uu'$, so $utt'u'$ can
fire and leads to the same marking as $tt'$. The incidence matrix of the net is shown on the right}
\end{figure}

\begin{defi}[T-invariant]
Let $N=(S,T,F)$ be a Petri net and $\I$ its incidence matrix. A vector $x\in \nat^T$ is called a
{\em $T$-invariant} if $\I x=0$. If a T-invariant corresponds to some executable firing sequence, it is called
realizable.
\end{defi}

A realizable $T$-invariant represents a cycle in the state space. Corresponding firing sequences do not change the marking. However, its interleaving with another sequence $\sigma$
may turn  $\sigma$ from unrealizable to realizable. 

Solving the state equation is a
non-negative integer programming problem. From linear algebra we know that the solution space is 
semi-linear.

\begin{cor}[Solution space]
For a given state equation $m+\I x=m'$ over a net $N=(S,T,F)$, there are numbers $j,k\in\nat$ and finite sets of vectors 
$B=\{b_i\in\nat^T\,|\,1\le i\le j\}$ (base vectors) and $P=\{p_i\in\nat^T\,|\,1\le i\le k\}$ (period vectors) such that:
\begin{iteMize}{$\bullet$}
\item all $b_i\in B$ are pairwise incomparable (by standard componentwise comparison for vectors) and thus minimal solutions,
\item $P$ forms a basis for the non-negative solution space $P^*=\{\sum_{i=1}^kn_ip_i\,|\,n_i\in\nat, p_i\in P\}$ of $\I x=0$,
\item for all solutions $x$ there are $n_i\in\nat$ for $1\le i\le k$ and $n\in\{1,\ldots,j\}$ such that
	$x = b_n+\sum_{i=1}^kn_ip_i$,
\item for every solution $x$, all vectors of the set $x+P^*$ are solutions as well.
\end{iteMize}
\end{cor}

Note that only linear combinations with nonnegative coefficients are considered in this representation.
In this setting, the number of base vectors as well as the number of period vectors may exponentially depend on the
size of the net. Permitting negative combinations (and thus solutions in the integers instead of the natural numbers)
would yield a significant loss in precision of the state equation.

So we know that all solutions can be obtained by taking a minimal solution $b$ of the state equation
and adding a linear combination of $T$-invariants from some basis $P$. Usually, not all the elements
from $B$ and $P$ we use for a solution are realizable, though. While the sum of two realizable $T$-invariants
remains realizable (just concatenate the according firing sequences as they have identical initial and final marking), 
the sum of two non-realizable
$T$-invariants may well become realizable. This can be seen in Fig.~\ref{f.tinv2}, where neither $\parikh(tt')$
nor $\parikh(uu')$ is realizable under the marking $m$ with $m(s_1)=m(s_4)=1$ and $m(s_2)=m(s_3)=0$, 
but the sequence $tut'u'$ realizes $\parikh(tt'uu')$. The matter is even more complicated when a minimal
solution from $B$ is introduced, because positive minimal solutions are never $T$-invariants (unless $m=m'$), i.e.\ 
they change the marking of the net, so their realizations cannot just be concatenated.

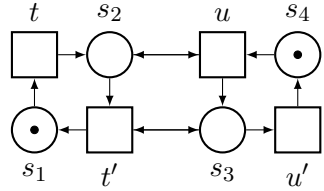
\begin{figure}[tb]
\centering
\begin{tikzpicture}[scale=0.5]
\node[transition,label=above:$t$] (t) {};
\node[place,label=above:$s_2$] (s2) [right of=t] {};
\node[transition,label=below:$t'$] (tx) [below of=s2] {};
\node[place,label=below:$s_1$] (s1) [below of=t] {};
\node[transition,label=above:$u$] (u) [right of=s2,xshift=5mm] {};
\node[place,label=above:$s_4$] (s4) [right of=u] {};
\node[transition,label=below:$u'$] (ux) [below of=s4] {};
\node[place,label=below:$s_3$] (s3) [below of=u] {};
\draw (s1) node[token] {};
\draw (s4) node[token] {};
\draw[arrow] (s1) to (t);
\draw[arrow] (t) to (s2);
\draw[arrow] (s2) to (tx);
\draw[arrow] (tx) to (s1);
\draw[arrow] (s4) to (u);
\draw[arrow] (u) to (s3);
\draw[arrow] (s3) to (ux);
\draw[arrow] (ux) to (s4);
\draw[arrow] (s2) to (u);
\draw[arrow] (u) to (s2);
\draw[arrow] (s3) to (tx);
\draw[arrow] (tx) to (s3);
\end{tikzpicture}
\caption{\label{f.tinv2}Neither the $T$-invariant $\parikh(tt')$ nor $\parikh(uu')$ is realizable, but $\parikh(tt'uu')$ is, 
	by the sequence $tut'u'$}
\end{figure}

\section{Traversing the Solution Space}\label{sec3}

Though realizability of transition vectors is and remains a problem, 
the first real problem we encounter when we try to solve the
state equation is a practical one. While there are solvers that can determine the complete sets of base
and period vectors for the solution space, e.g. {\em 4ti2}~\cite{4ti2}, these programs can only do that for very
small systems. If the system has a hundred or more variables, we are out of luck. On the other hand,
integer programming (IP) solvers like {\em lp\_solve}~\cite{lpsolve} are much faster but will only find one solution 
to the state equation at a time. 

Another point of interest to know about IP solvers is the objective
function, which can be any function over the variables of the linear system to solve. 
While looking for a solution, the solver tries to minimize or maximize this function. Minimizing seems more
valuable here, since we might use the sum over all variables as our objective function, effectively telling
the solver to produce a solution that would lead to a shortest firing sequence if realizable.
In the following, we will assume such an objective function.

Fortunately, we can force an IP solver to produce more than just one solution --- this is the CEGAR part of our approach. 
If a solution found is not realizable, we may add an inequation to our state
equation to forbid that solution. Starting the IP solver again will then lead to a different solution. The trick is,
of course, to add inequations in such a way that no realizable solution is lost. 

\begin{defi}[Constraints]
Let $N=(S,T,F)$ be a Petri net.
We define two forms of constraints, both being linear inequations over transitions:
\begin{iteMize}{$\bullet$}
\item a {\em jump constraint} takes the form $t < n$ with $n\in\nat$ and $t\in T$. 
In general, it is intended to switch (jump)
to another base solution, exploiting the incomparability of different minimal base solutions.
\item an {\em increment constraint} takes the form $\sum_{i=1}^kn_it_i \ge n$ with $n_i\in\Z$, $n\in\nat$, and $t_i\in T$.
Among others, it can be used to force non-minimal solutions.
\end{iteMize}
\end{defi}

\noindent To understand the idea for differentiating between these two forms of constraints, it is necessary to introduce
the concept of a partial solution first. A partial solution is obtained from a solution of the state equation under
given constraints by firing as many transitions as possible.

\begin{defi}[Partial solution]
Let $N=(S,T,F)$ be a Petri net and $\ord$ a total order over $\nat^T$ that includes the partial order given by
$x<y$ if $\sum_{t\in T}x(t)<\sum_{t\in T}y(t)$. 

\noindent A {\em partial solution} of a reachability problem 
$(N,m,m')$ is a 
tuple 
$(\C,x,\sigma,r)$ 
of
\begin{iteMize}{$\bullet$}
\item a family of (jump and increment) constraints $\C=(c_1,\ldots,c_n)$,
\item the $\ord$-smallest solution $x$ fulfilling the state equation of $(N,m,m')$ {\em and} the constraints of $\C$,
\item a firing sequence $\sigma\in T^*$ with $m\step{\sigma}$ and $\parikh(\sigma)\le x$,
\item a remainder $r$ with $r = x-\parikh(\sigma)$ and $\forall t\in T$: $(r(t)>0\then\neg m\step{\sigma t})$.
\end{iteMize}
The vectors $x$ and $r$ are included for convenience only, they can be computed from $\C$, $\sigma$, $\ord$, and the problem instance.

A {\em full solution} is a partial solution $(\C,x,\sigma,r)$ with $r=0$. In this case, $\sigma$ is a firing
sequence solving our reachability problem (with answer 'yes').
\end{defi}

We choose $\ord$ such that an IP solver can be assumed to always 
produce the $\ord$-smallest solution that does not contradict its linear system of equations.
Note that from any firing sequence solving our reachability problem we can easily deduce a full solution:

\begin{cor}[Realizable solutions are full solutions]
For any realizable solution $x$ of the state equation (realized by a firing sequence $\sigma$)
we find a full solution $(\C,x,\sigma,0)$
where $\C$ consists of constraints $t\geq x(t)$ for every $t$ with $x(t)>0$, and $\parikh(\sigma)=x$.
\end{cor}
Note, that $x$ is the smallest solution fulfilling $\Gamma$ and therefore also the $\ord$-smallest solution.

By adding a constraint to a partial solution we may obtain new partial solutions (or not, if the linear
system becomes infeasible). Any full solution can eventually be reached by consecutively extending an $\ord$-minimal 
partial solution with constraints. The following lemma is thus a core argument for the correctness of our approach.

\begin{lem}[A path to a full solution]\label{L.PTFL}
Let $b$ be the $\ord$-minimal solution of the state equation of a reachability problem $(N,m,m')$ and 
$ps'=((c_j)_{1\le j\le\ell},b'+\sum_{i=1}^kn_ip_i,\sigma',0)$ a full solution of the problem. 
For $0\le n\le\ell$, there are partial solutions $ps_n=((c_j)_{1\le j\le n},x_n,\sigma_n,r_n)$
with $ps_0=(\emptyset,b,\sigma_0,r_0)$, $ps_\ell=ps'$, and $x_{n_1}\le_{\ord} x_{n_2}$ for $n_1\le n_2$.
\end{lem}
\begin{proof}
Let $\C_n=(c_j)_{1\le j\le n}$.
If $ps_{n_1},ps_{n_2}$ are two partial solutions (with $1\le n_1<n_2\le\ell$) then $x_{n_2}$ is a solution
of the state equation plus $\C_{n_1}$, since it even fulfills the state equation plus $\C_{n_2}$ with $\C_{n_1}\subseteq \C_{n_2}$.
As $x_{n_1}$ is the $\ord$-smallest solution of the state equation plus $\C_{n_1}$, $x_{n_1}\le_{\ord} x_{n_2}$ holds.
Therefore, $b\le_{\ord} x_1\le_{\ord} \ldots \le_{\ord} x_\ell$. Since $x_\ell=b'+\sum_{i=1}^kn_ip_i$ is an existing solution of the
strictest system, i.e. state equation plus $\C_\ell$, each system of state equation plus one family of constraints $\C_n$
is solvable. As a $\sigma_n$ can be determined by just firing transitions as long as possible, all the partial
solutions $ps_n$ exist.
\end{proof}

Now, let us assume a partial solution $ps=(\C,x,\sigma,r)$ that is not a full solution, i.e. $r\neq 0$.
Obviously, some transitions cannot fire often enough. There are three possible remedies for this situation:
\begin{enumerate}[(1)]
\item $x$ may still be realizable by a different firing sequence that corresponds to $x$. That is, we can find a full solution $ps'=(\C,x,\sigma',0)$ with $\parikh(\sigma')=x$.
\item We can add a jump constraint to obtain an $\ord$-greater solution vector for a different partial solution.
\item If $r(t)>0$ for some transition $t$, we can add an increment constraint to increase the maximal number of tokens available on
	a place in the preset of $t$. Since the final marking remains the same, this means to borrow tokens for such a place. This
	can be done by adding a $T$-invariant containing the place to the solution. 
\end{enumerate}\smallskip

\noindent A visualization of these ideas can be seen in Fig.~\ref{f.solpath} where $b$ denotes the $\ord$-smallest solution.
The cone over $b$ represents all solutions $b+P^*$ with $P$ being the set of period vectors, i.e. $T$-invariants.
Jump constraints lead along the dashed or dotted lines to the next $\ord$-minimal solution while normal arrows representing
increment constraints lead upwards to show the addition of a $T$-invariant. How to build constraints doing just
what we want them to do is the content of the next section.

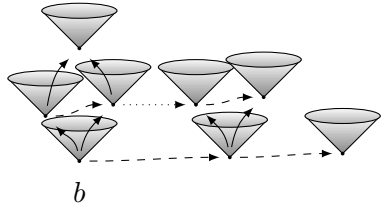
\begin{figure}[tb]
\centering
\begin{tikzpicture}[scale=0.5]
\shade[top color=black!1, bottom color=black!50] (1,1)--(0,0)--(-1,1)--cycle;
\draw(-1,1)--(0,0)--(1,1);
\shadedraw[top color=black!1, bottom color=black!30] (0,1) ellipse (1 and 0.2);
\fill[color=black] (0,0) circle (0.05);
\draw(0,0) node[label=below:$b$] (c1) {};

\shade[top color=black!1, bottom color=black!50] (5,1.1)--(4,0.1)--(3,1.1)--cycle;
\draw(3,1.1)--(4,0.1)--(5,1.1);
\shadedraw[top color=black!1, bottom color=black!30] (4,1.1) ellipse (1 and 0.2);
\fill[color=black] (4,0.1) circle (0.05);
\draw(4,0.1) node (c2) {};

\shade[top color=black!1, bottom color=black!50] (8,1.2)--(7,0.2)--(6,1.2)--cycle;
\draw(6,1.2)--(7,0.2)--(8,1.2);
\shadedraw[top color=black!1, bottom color=black!30] (7,1.2) ellipse (1 and 0.2);
\fill[color=black] (7,0.2) circle (0.05);
\draw(7,0.2) node (c3) {};

\shade[top color=black!1, bottom color=black!50] (0.1,2.2)--(-0.9,1.2)--(-1.9,2.2)--cycle;
\draw(-1.9,2.2)--(-0.9,1.2)--(0.1,2.2);
\shadedraw[top color=black!1, bottom color=black!30] (-0.9,2.2) ellipse (1 and 0.2);
\fill[color=black] (-0.9,1.2) circle (0.05);
\draw(-0.9,1.2) node (c11) {};

\shade[top color=black!1, bottom color=black!50] (1.9,2.5)--(0.9,1.5)--(-0.1,2.5)--cycle;
\draw(-0.1,2.5)--(0.9,1.5)--(1.9,2.5);
\shadedraw[top color=black!1, bottom color=black!30] (0.9,2.5) ellipse (1 and 0.2);
\fill[color=black] (0.9,1.5) circle (0.05);
\draw(0.9,1.5) node (c12) {};

\shade[top color=black!1, bottom color=black!50] (4.1,2.5)--(3.1,1.5)--(2.1,2.5)--cycle;
\draw(2.1,2.5)--(3.1,1.5)--(4.1,2.5);
\shadedraw[top color=black!1, bottom color=black!30] (3.1,2.5) ellipse (1 and 0.2);
\fill[color=black] (3.1,1.5) circle (0.05);
\draw(3.1,1.5) node (c21) {};

\shade[top color=black!1, bottom color=black!50] (5.9,2.7)--(4.9,1.7)--(3.9,2.7)--cycle;
\draw(3.9,2.7)--(4.9,1.7)--(5.9,2.7);
\shadedraw[top color=black!1, bottom color=black!30] (4.9,2.7) ellipse (1 and 0.2);
\fill[color=black] (4.9,1.7) circle (0.05);
\draw(4.9,1.7) node (c22) {};

\shade[top color=black!1, bottom color=black!50] (1,4)--(0,3)--(-1,4)--cycle;
\draw(-1,4)--(0,3)--(1,4);
\shadedraw[top color=black!1, bottom color=black!30] (0,4) ellipse (1 and 0.2);
\fill[color=black] (0,3) circle (0.05);
\draw(0,3) node (c111) {};

\draw[arrow] (c1) to[out=100,in=315] (c11);
\draw[arrow] (c1) to[out=80,in=225] (c12);
\draw[arrow] (c11) to[out=80,in=225] (c111);
\draw[arrow] (c12) to[out=100,in=315] (c111);
\draw[arrow] (c2) to[out=100,in=315] (c21);
\draw[arrow] (c2) to[out=80,in=225] (c22);
\draw[arrow,dashed] (c1) to[out=0,in=180] (c2);
\draw[arrow,dashed] (c2) to[out=0,in=180] (c3);
\draw[arrow,dashed] (c11) to[out=0,in=180] (c12);
\draw[arrow,dotted] (c12) to[out=0,in=180] (c21);
\draw[arrow,dashed] (c21) to[out=0,in=180] (c22);
\end{tikzpicture}
\caption{\label{f.solpath}Paths from the $\ord$-minimal solution $b$ to any solution. Black dots represent solutions, cones
stand for linear solution spaces over such solutions, which may or may not intersect or include each other.
Normal arrows increment a solution by adding a $T$-invariant, dashed arrows are jumps to an incomparable $\ord$-greater solution.
Such jumps can also occur on higher levels of linear solution spaces, shown by the dotted arrow}
\end{figure}

\section{Building Constraints}\label{sec4}

Let us first argue that for a state equation, any of the minimal solution vectors in $B$ can be obtained by using jump
constraints.

\begin{lem}[Jumps to minimal solutions]\label{L.BJ}
Let $b,b'\in B$ be base vectors of the solution space of the state equation $m+\I x=m'$ plus some set of constraints $\C$. 
Assume $b$ to be the $\ord$-minimal solution of the system. 
Then, we can obtain $b'$ as output of our IP solver by consecutively adding jump constraints of the form $t_i < n_i$ with $n_i\in\nat$ to $\C$.
\end{lem}
\begin{proof}
We know $b\le_{\ord} b'$ holds, but since $b'$ is a minimal solution, $b\le b'$ cannot hold. Therefore, a transition $t$ with
$b'(t)<b(t)$ must exist. After adding the constraint $t<b(t)$ to $\C$
the IP solver can no longer generate $b$ as a solution. Assume $b''$ is the newly generated solution. If $b'=b''$
we are done. Otherwise, since $b'$ fulfills $t<b(t)$, it is still a solution of our system, and also a minimal one as the
solution space is restricted by the added constraint. Thus, $b''\le_{\ord} b'$ holds and
we may recursively use the same argument as above for $b:=b''$. Since there are only finitely many solutions $\ord$-smaller
than $b'$, the argument must terminate reaching $b'$.
\end{proof}

Non-minimal solutions may not be reachable this way, since the argument ``$b'(t)<b(t)$ for some $t$'' does not necessarily hold.
We will need increment constraints for this, but unluckily, increment constraints and jump constraints may 
contradict each other. Assume our state equation has a solution of the form $b'+p$ with a period vector $p\in P$
and to obtain $b'\in B$ from the $\ord$-minimal solution $b\in B$ we need to add (at least) a jump constraint $t_i<n_i$ to the state equation.
If $p$ contains $t_i$ often enough, we will find that $(b'+p)(t_i)\ge n_i$ holds. Therefore, $b'+p$ is not a solution of the
state equation plus the constraint $t_i<n_i$, i.e. adding an increment constraint demanding enough occurrences of $t_i$
for $b'+p$ will render the linear equation system infeasible. The only way to avoid this problem is to remove the jump
constraints before adding increment constraints.

\begin{lem}[Transforming jumps]
Let $z$ be the $\ord$-minimal solution of the state equation $m+\I x=m'$ plus some constraints $\C$. Let $\C'$ consist of all
increment constraints of $\C$ plus a constraint $t\ge z(t)$ for each transition $t$. Then, for all $y\ge z$,
$y$ is a solution of $m+\I x=m'$ plus $\C\cap\C'$ if and only if $y$ is a solution of $m+\I x=m'$ plus $\C'$.
Furthermore, no $\ord$-smaller solution of $m+\I x=m'$ plus $\C$ than $z$ solves $m+\I x=m'$ plus $\C'$.
\end{lem}
\begin{proof}
Let $y\ge z$ be a solution of $m+\I x=m'$ plus $\C\cap\C'$. The additional constraints in $\C'$ only demand $y(t)\ge z(t)$, which
is obviously the case. The other direction is trivial. For the second part, let $z'\le_{\ord} z$ with $z\not=z'$ be some solution of
$m+\I x=m'$ plus $\C$. Since $\sum_tz'(t)\le\sum_tz(t)$ (following from $\ord$) but $z\not=z'$, for at least one transition $t$ holds $z'(t)<z(t)$. 
Due to the constraint $t\ge z(t)$ in $\C'$, $z'$ cannot be a solution of $m+\I x=m'$ plus $\C'$.
\end{proof}

As a consequence, if we are only interested in solutions of the cone $z+P^*$ over $z$, we can add increment constraints guaranteeing
solutions greater or equal than $z$ and remove all jump constraints without any further restriction. Our IP solver will yield $z$ as the $\ord$-minimal
solution for both families of constraints, $\C$ and $\C'$, and we can add further constraints leading us to any solution in the
cone $z+P^*$ now.

Let $ps=(\C,x,\sigma,r)$ now be a partial solution with $r>0$. We would like to determine sets of places that need additional tokens
(and the number of these tokens) that would enable us to fire the remainder $r$ of transitions. Obviously, this problem is harder
than the original problem of finding out if a transition vector is realizable, i.e. just testing if zero additional tokens are
sufficient. A recursive approach would probably be very inefficient as for every solution $x$ there may be many different remainders $r$.
Even though the remainders are smaller than the solution vector $x$, the number of recursion steps might easily grow exponentially with
the size of $x$, i.e. $\sum_tx(t)$. We therefore adopt a different strategy, namely finding good heuristics to estimate the number of
tokens needed. If a set of places actually needs $n$ additional tokens with $n>0$, our estimate may be any number from one to $n$.
If we guess too low, we will obtain a new partial solution allowing us to make a guess once again, (more or less) slowly approaching
the correct number. We propose a two-part algorithm.The first part computes sets of places and transitions that are of interest, the second
estimates the number of tokens. For the first part, we use a dependency graph that is known from partial order reduction approaches \cite{valmari2009}.
For every disabled transition, we can choose an insufficiently marked place. For every such place, we consider its pre-transitions. Applying this
idea to a partial solution (where all transitions in $r$ are disabled), this graph yields cycles, or more generally strongly connected components) of mutually blocked 
transitions and their insufficiently marked scapegoat places. In order to make parts of $r$ fireable, we need to interleave $r$ with a T-invariant that
is able to lend tokens to any of the involved scapegoat places. Obviously, source SCC (i.e.~SCC without incoming edges) are of particular
interest for enabling parts of $r$. The following algorithm computes sets of places where additional tokens are necessary.

\smallskip{\small\noindent
{\tt
{\bf input}: Reachability prob. $(N,m,m')$; partial solution $ps=(\C,x,\sigma,r)$\\
{\bf output}: A set of tuples $(S_i,T_i,X_i)$ with $S_i\subseteq S$, $T_i\cup X_i\subseteq T$\\
Determine $\hat{m}$ with $m\step{\sigma}\hat{m}$;\\
Build a bipartite graph $G=(S_0\cup T_0,E)$ with\\
$T_0:=\{t\in T\,|\,r(t)>0\}$; $S_0:=\{s\in S\,|\,\exists t\in T_0$: $F(s,t)>\hat{m}(s)\}$;\\
$E:=\{(s,t)\in S_0\times T_0\,|\,F(s,t)>\hat{m}(s)\}\cup\{(t,s)\in T_0\times S_0\,|\,F(t,s)>F(s,t)\}$;\\
Calculate the strongly connected components ($SCCs$) of $G$;\\
$i:=1$;\\
{\bf for each} source $SCC$ (i.e.~one without incoming edges):\\
\hspace*{3mm} $S_i := SCC \cap S_0$;\\
\hspace*{3mm} $T_i := SCC \cap T_0$;\\
\hspace*{3mm} $X_i := \{t\in T_0\backslash SCC\,|\,\exists s\in S_i:\;(s,t)\in E\}$;\\
\hspace*{3mm} $i := i+1$;\\
{\bf end for}
}}\smallskip

The edges of the graph $G$ constructed in the algorithm have a different meaning depending on their direction.
Edges from transitions to places signal that the transition would increase the number of tokens on the place
upon firing, while edges in the other direction show the reason for the non-enabledness of the transition.
A source SCC, i.e. a strongly connected component without incoming edges from other components, can therefore
not obtain tokens by the firing of transitions from other SCCs. This means, tokens must come from somewhere else, 
that is, from firing transitions
not appearing in the remainder $r$. For each set of places $S_i$ such identified as non-markable by the remainder itself,
there are two sets of transitions. If one transition from the set $T_i$ would become firable, it is possible that
all other transitions could fire as well, since the former transition effectively produces tokens on some place in
the component. If the set $T_i$ is empty (the SCC consisting of a single place), the transitions mentioned in $r$
will not produce any tokens on the SCC. Thus, the token needs of the transitions depending on this SCC, i.e.\
those in $X_i$, must all be fulfilled together, since they cannot activate each other. Overall, we obtain:
\begin{lem}\label{L.ALG1}
The previous algorithm determines source SCCs $(S_i,T_i)$ with insufficient numbers of tokens to enable the
remainder $r$ of a partial solution as well as sets $X_i$ of transitions depending on those SCCs for firing.
\end{lem}

Since enough additional tokens to enable one transition from a set $T_i$ might later enable the rest of $T_i$ and
$X_i$ as well, it is difficult to compute the exact number of tokens needed on some SCC and where to place them.
The following algorithm thus is a heuristic to determine a least number of tokens necessary on a source SCC.

\smallskip{\small\noindent
{\tt
{\bf input}: A tuple $(S_i,T_i,X_i)$; $(N,m,m')$ and $\hat{m}$ from above\\
{\bf output}: A number of tokens $n$ (additionally needed for $S_i$)\\
{\bf if} $T_i\neq\emptyset$\\
{\bf then} $n:=\min_{t\in T_i}(\sum_{s\in S_i}\max\{0,(F(s,t)-\hat{m}(s))\})$\\
{\bf else} sort $X_i$ in groups $G_j:=\{t\in X_i\,|\,F(t,s)=j\}$ (with $S_i=\{s\}$);\\
\hspace*{7mm} $n:=0$; $c:=0$;\\
\hspace*{7mm} {\bf for} $j$ {\bf with} $G_j\neq\emptyset$ {\bf downwards loop}\\ 
\hspace*{14mm} $c:=c+j+\sum_{t\in G_j}r(t)*(F(s,t)-j)$;\\
\hspace*{14mm} {\bf if} $c>0$ {\bf then} $n:=n+c$ {\bf end if};\\
\hspace*{14mm} $c:=-j$\\
\hspace*{7mm} {\bf end for}\\
{\bf end if}
}}\smallskip

\begin{lem}\label{L.ALG2}
For each set of places $S_i$ that need additional tokens according to the first part of the algorithm, the second
part estimates that number of tokens (in a range from one to the actual minimum number of tokens necessary).
\end{lem}
\begin{proof}
Consider a triple ($S_i$,$T_i$,$X_i$) computed by the first algorithm.
If $T_i$ is not empty, only line~4 of the second algorithm is executed, computes the number of tokens missing
for each of the transitions in $T_i$, and takes the minimum over these numbers. By construction, this number is
at least one and any lower number would not activate any transition from $T_i$. 

If $T_i$ is empty, all transitions in $X_i$ depend on a single place which must provide all the
necessary tokens to fire all the transitions in $X_i$. 
Note that while the transitions in $X_i$ all effectively consume tokens from $s\in S_i$, they may also put tokens back onto this
place due to a loop. By firing those transitions with the lowest $F(t,s)$-values last, we minimize the leftover.
Transitions with the same $F(t,s)$-value $j$ can be processed together, each consuming effectively $F(s,t)-j$ tokens
except for the ``first'' transition which requires the existence of $j$ additional tokens. 
If some group $G_j$ of transitions leaves tokens
on $s$, the next group can consume them, which is memorized in the variable $c$ (for carryover or consumption).
Overall, we get the minimal number of tokens necessary to fire all transitions in $X_i$.

Observe, that the algorithm cannot return zero or negative values: 
There must be at least one transition in $T_i\cup X_i$, otherwise
there would be no transition that cannot fire due to a place in $S_i$ and the places in $S_i$ would not have been
computed at all. If $T_i$ is not empty, line~4 in the algorithm minimizes over positive values, i.e.\ the
numbers of tokens missing for each transition in $T_i$; if $T_i$
is empty, line~8 will set $c$ to a positive value at its first execution, yielding a positive value for $n$. 
\end{proof}

We can thus try to construct a constraint from a set of places $S_i$ generated by the first part of the algorithm
and the token number calculated in the second part. Since our state equation has transitions as variables, we
must transform our condition on places into one on transitions first. 

\begin{lem}\label{L.CON}
Let $N=(S,T,F)$ be a Petri net, $(N,m,m')$ the reachability problem to be solved, $ps=(\C,x,\sigma,r)$ a partial solution
with $r>0$, and $\hat{m}$ the marking reached by $m\step{\sigma}\hat{m}$.
Let $S_i$ be a set of places and $n$ a number of tokens to be generated on $S_i$.
Further, let $T_i:=\{t\in T\,|\,r(t)=0\,\wedge\,\sum_{s\in S_i}(F(t,s)-F(s,t))>0\}$.
We define a constraint $c$ by
\[ \sum_{t\in T_i} \sum_{s\in S_i} (F(t,s)-F(s,t))t \ge n+\sum_{t\in T_i}\sum_{s\in S_i}(F(t,s)-F(s,t))\parikh(\sigma)(t).\]
Then, for the system $m+\I x=m'$ plus $\C$ plus $c$, if our IP solver can generate a solution $x+y$ ($y$ being a $T$-invariant) 
we can obtain a partial solution $ps'=(\C\cup\{c\},x+y,\sigma\tau,r+z)$ with $\parikh(\tau)+z=y$. Furthermore,
$\sum_{t\in T}\sum_{s\in S_i}(F(t,s)-F(s,t))y(t)\ge n$.
\end{lem}
\begin{proof} {\em (Sketch)}
First, note that $T_i$ contains the transitions that produce more on $S_i$ than they consume, but we have explicitly excluded
all transitions of the remainder $r$, since we do not want the IP solver to increase the token production on $S_i$ by adding
transitions that could not fire anyway. I.e., we would like to have a chance to fire the additional transitions in $y$ at
some point, though there are no guarantees. The left hand side of $c$ contains one instance of a transition $t$ for each
token that $t$ effectively adds to $S_i$. If we apply some transition vector $u$ to the left hand side of $c$ (i.e.\ 
replacing $t$ by $u(t)$ for each transition $t\in T_i$), we therefore
get the number of tokens added to $S_i$ by firing the transitions from $T_i$ in $u$. Of course, other transitions in $u$ (outside $T_i$) might
reduce this number again. For the right hand side of $c$, we calculate how many tokens are actually added to $S_i$ by the
transitions from $T_i$ in the firing sequence $\sigma$ (and therefore also in the solution $x$) and increase that number by 
the $n$ extra tokens we would like to have. Since the extra tokens cannot come from $x$ in a solution $u:=x+y$, they must
be produced by $y$, i.e. $\sum_{t\in T}\sum_{s\in S_i}(F(t,s)-F(s,t))y(t)\ge n$. We might be able to fire some portion of $y$
after $\sigma$, resulting in the obvious $\parikh(\tau)+z=y$.
\end{proof}

When we apply our constraint we might get less or even more than the $n$ extra tokens, depending on the $T$-invariants in the net.
There are three possible outcomes when we apply the new constraint: The constraint might be unsatisfiable, we might detect
that the new constraint has not brought us any closer to a solution, or some of the remainder
transitions can now fire (or have at least come closer to firing). In the first two cases, our partial solution just cannot 
be extended to a full solution, nothing is lost if we throw it away. Failing to detect the second case will not cut off
any solutions but might prohibit termination. In the third case, we continue extending the partial
solution with the help of further constraints. Therefore, if a solution exists, we can still find it, and we know that we
can find it with the help of jump and increment constraints, since we only add necessary constraints. Further jump
constraints may be necessary as there may be many incomparable, minimal solutions fulfilling the latest increment constraint.

\begin{thm}[Reachability of solutions]
If a reachability problem has a solution, a realizable solution of the solution space of the state equation can be reached by consecutively adding constraints
to the system of equations, always transforming jump constraints before adding increment constraints.
This even holds, given some solution $x$, if we limit the increment constraints to those built by lemma~\ref{L.CON} and jump constraints to
those of the form $t<x(t)$.
\end{thm}
\begin{proof} {\em (Sketch)}
The first sentence is obvious from lemma~\ref{L.PTFL}. For the second part, let $\sigma$ be a realizable solution
with $\parikh(\sigma)=x$ and let $b$ be a minimal solution with $b\le x$. We can obtain $b$ by a set of jump
constraints according to lemma~\ref{L.BJ}. Either $b$ is realizable and we are done, or $b$ is not realizable.
In the latter case, we compute an increment constraint forcing additional tokens to undermarked places according
to lemma~\ref{L.ALG1}, \ref{L.ALG2}, and~\ref{L.CON}, not losing any full solutions. Let $y$ be the solution
computed for this extended system, then $b\le y$ holds. If $y$ is realizable, we are done. If $y\le x$, we add the next
increment constraint. Otherwise, we test further jump constraints now and obtain all solutions $y_1,\ldots,y_n$ 
fulfilling the increment constraint but being incomparable to $y$. For one of them, $y_k\le x$ must hold (as we
have not lost any full solutions). We continue by checking realizability of $y_k$ and, in the negative case,
repeat adding a necessary increment constraint. Since $x$ is finite and each constraint leads only to
bigger solutions, at some point we must find a realizable solution or reach $x$.
\end{proof}

\section{Finding Partial Solutions}\label{sec5}

Producing partial solutions $ps=(\C,x,\sigma,r)$ from a solution $x$ of the state equation (plus $\C$) is actually
quite easily done by brute force. We can build a tree with marking-annotated nodes and the firing of transitions as edges,
allowing at most $x(t)$ instances of a transition $t$ on any path from the root of the tree to a leaf. Any leaf is
a new partial solution from which we may generate new solutions by adding constraints to the state equation and
forwarding the evolving linear system to our IP solver. If we just make a depth-first-search through our tree and
backtrack at any leaf, we build up all possible firing sequences realizable from $x$. This is obviously possible
without explicitly building the whole tree at once, thus saving memory.
Of course, the tree might grow exponentially in the size of the solution vector $x$ and so some optimizations
are in order to reduce the run-time. We would like to suggest a few ones here, especially partial order reductions.

\begin{enumerate}[(1)]
\item The stubborn set method \cite{ksv06} determines a set of transitions that can be fired before all others by
investigating conflicts and dependencies between transitions at the active marking. The stubborn set is often much smaller than
the set of enabled transitions under the same marking, leading to a tree with a lower degree. In our case, in particular the version
of \cite{schmidt_atpn99} is useful as, using this method, the reduced state space contains, for each trace to the target marking,
at least one permution of the same trace. Hence, the reduction is consistent with the given solution of the state equation.
\item Especially if transitions should fire multiple times ($x(t)>1$) we observe that the stubborn set
method alone is not efficient. The situation in Fig.~\ref{f.tree} may occur quite often. Assume we reach some marking
$\hat{m}$ by a firing sequence $\alpha$, so that transitions $t$ and $u$ are enabled. After proceeding through the
subtree behind $t$ we backtrack to the same point and now fire $u$ followed by some sequence $\sigma$ after which $t$
is enabled, leading to $m\step{\alpha}\hat{m}\step{u\sigma t}\widetilde{m}$. If $\hat{m}\step{t\sigma u}$ holds, we know that
it reaches the same marking $\widetilde{m}$ and the same remainder $r$ of transitions still has to fire. Therefore, in
both cases the future is identical. Since we have already investigated what happens after firing $\alpha t\sigma u$,
we may backtrack now omitting the subtree after $\alpha u\sigma t$. Note that a test if $\hat{m}\step{t\sigma u}$ holds
is quite cheap, as only those places $s$ with $\I_{s,t}<\I_{s,u}$ can prevent the sequence $t\sigma$.
Enabledness of $u$ after $t\sigma$ can be tested by reverse calculating $\widetilde{m}_1=\widetilde{m}-\I u$ and 
checking whether $\widetilde{m}_1$ is a marking and $\widetilde{m}_1\step{u}\widetilde{m}$ holds. 
\item There are situations where a leaf belongs to a partial solution $ps'$ that cannot lead to a (new) full solution. In this
case the partial solution does not need to be processed. If we already tried to realize $x$ yielding a partial solution
$ps=(\C,x,\sigma,r)$ and $ps'=(\C\cup\{c\},x+y,\sigma,r+y)$ is our new partial solution with an increment constraint $c$
and a $T$-invariant $y$, any realizable solution $x+y+z$ obtainable from $ps'$ can also be reached from $ps$ by first
adding a constraint $c'$ for the $T$-invariant $z$ (and later $c$, $y$). If no transition of $z$ can be fired after $\sigma$,
$y+z$ is also not realizable after firing $\sigma$. We may be able to
mingle the realization of $z$ with the firing of $\sigma$, but that will be reflected by alternate partial solutions
(compared to both, $ps$ and $ps'$). Therefore, not processing $ps'$ will not lose any full solutions.
\item A similar situation occurs for $ps'=(\C\cup\{c\},x+y,\sigma\tau,r)$ with $\parikh(\tau)=y$. There is one problem,
though. Since we estimated a token need when choosing $c$ and that estimate may be too low, it is possible that while
firing $\tau$ we get closer to enabling some transition $t$ in $r$ without actually reaching that limit where $t$
becomes firable. We thus have to check for such a situation (by counting the minimal number of missing tokens for
firing $t$ in the intermediate markings occurring when firing $\sigma$ and $\tau$). If $\tau$ does not help in
approaching enabledness of some $t$ in $r$, we do not need to process $ps'$ any further.
\item Partial solutions should be memorized if possible to avoid using them as input for CEGAR again if they show up more than once.
\end{enumerate}

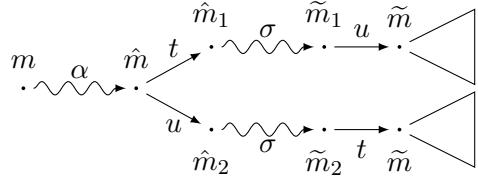
\begin{figure}[tb]
\centering
\begin{tikzpicture}[scale=0.5]
\fill (0,3) circle (0.05);
\draw (0,3) node[label=above:$m$] (root) {};
\fill (3,3) circle (0.05);
\draw (3,3) node[label=above:$\hat{m}$] (hm) {};
\draw[arrow,snake=snake] (root) -- node[label=above:$\alpha$] {} (hm);
\fill (5,4.1) circle (0.05);
\draw (5,4.1) node[label=above:$\hat{m}_1$] (hm1) {};
\fill (5,1.9) circle (0.05);
\draw (5,1.9) node[label=below:$\hat{m}_2$] (hm2) {};
\draw[arrow] (hm) -- node[label=above:$t$] {} (hm1);
\draw[arrow] (hm) -- node[label=below:$u$] {} (hm2);
\fill (8,4.1) circle (0.05);
\draw (8,4.1) node[label=above:$\widetilde{m}_1$] (tm1) {};
\fill (8,1.9) circle (0.05);
\draw (8,1.9) node[label=below:$\widetilde{m}_2$] (tm2) {};
\draw[arrow,snake=snake] (hm1) -- node[label=above:$\sigma$] {} (tm1);
\draw[arrow,snake=snake] (hm2) -- node[label=below:$\sigma$] {} (tm2);
\fill (10,4.1) circle (0.05);
\draw (10,4.1) node[label=above:$\widetilde{m}$] (tm3) {};
\fill (10,1.9) circle (0.05);
\draw (10,1.9) node[label=below:$\widetilde{m}$] (tm4) {};
\draw[arrow] (tm1) -- node[label=above:$u$] {} (tm3);
\draw[arrow] (tm2) -- node[label=below:$t$] {} (tm4);
\draw (tm3) -- +(2,1) -- +(2,-1) -- (tm3);
\draw (tm4) -- +(2,1) -- +(2,-1) -- (tm4);
\end{tikzpicture}
\caption{\label{f.tree}If both sequences $\alpha t\sigma u$ and $\alpha u\sigma t$ can be fired, the subtrees
after the nodes with marking $\tilde{m}$ are identical. Only one of the subtrees needs to be evaluated, the other one
may be omitted. Snaked lines denote firing sequences}
\end{figure}


\section{The complete algorithm}\label{sec6}

The following algorithm integrates the results from the previous sections, using a queue of
partial solutions as a job queue that is ordered by ascending size of the possible output, i.e.\ 
the firing sequence from an initial marking $m$ to a final marking $m'$.

\smallskip{\small\noindent
{\tt
{\bf input}: Reachability problem $(N,m,m')$\\
{\bf output}: Firing sequence from $m$ to $m'$ if one exists\\
{\bf var}: Queue $Q$ of partial solutions $(\Gamma,x,\sigma,r)$ to work on\\
{\bf var}: Set $S$ of all partial solutions computed (for optimisation (5))\\
Put $(\emptyset,0,\varepsilon,0)$ into $Q$;\\
{\bf while} $Q\not=\emptyset$:\\
\hspace*{3mm} Remove from $Q$ an element $(\Gamma,x,\sigma,r)$ with minimal size for $x$;\\ 
\hspace*{3mm} Compute new increment constraints $\Delta$ from $r$ (see section~\ref{sec4});\\
\hspace*{3mm} {\bf if} $(\Gamma\cup\Delta,\cdot,\cdot,\cdot)\in S$ {\bf then continue} to the next loop;\\ 
\hspace*{3mm} Compute the minimal solution $y$ for state equation plus $\Gamma\cup\Delta$;\\
\hspace*{3mm} {\bf if} no solution $y$ exists {\bf then continue} to the next loop;\\ 
\hspace*{3mm} {\bf for each} set of jump constraints $\emptyset\not=R\subseteq\{t<y(t)\,|\,y(t)>x(t)\}$\\
\hspace*{6mm} Transform jump constraints in $\Gamma\cup R$ to increment constraints;\\
\hspace*{6mm} Create a new partial solution $(\Gamma\cup R,x,\sigma,r)$ in $Q$ and $S$;\\
\hspace*{3mm} Traverse the tree of firing sequences $\sigma'$ with $\parikh(\sigma')\le y$:\\
\hspace*{6mm} (Depth first, use optimisations (1) and (2) to prune the tree)\\ 
\hspace*{6mm} Upon reaching a leaf ($\sigma'$ cannot be prolonged):\\
\hspace*{9mm} {\bf if} $y=\parikh(\sigma')$ {\bf then halt} with solution $\sigma'$;\\
\hspace*{9mm} {\bf if} optimisation (3) or (4) applies {\bf then} backtrack;\\ 
\hspace*{9mm} {\bf if} $(\Gamma\cup\Delta,y,\sigma',y-\parikh(\sigma'))\in S$ {\bf then} backtrack;\\
\hspace*{9mm} Put $(\Gamma\cup\Delta,y,\sigma',y-\parikh(\sigma'))$ into $Q$ and $S$\\
{\bf end while};\\
Print "no solution"
}}\smallskip

In each while-loop one candidate from the partial solutions queue is processed. First, the constraints $\Delta$
are added according to the non-firable remainder $r$ and a minimal solution vector $y$ is computed.
If this is successful, jump constraints are added to the old solution $x$ to obtain other new solutions
in a later loop. The minimal solution vector $y$ is then checked for realizable firing sequences.
If at some point a firing sequence is incomplete and not extendable, a new partial solution is
generated to be able to add more increment constraints later on. The result ``no solution''
can be obtained by two mechanisms: a partial solution is thrown away if no solution vector $y$
can be computed and the optimisations may prune the trees of firing sequences so much that no
new partial solutions are created anymore. There is no guarantee for termination, of course,
but in practise this works well.

Note that new jump constraints can lead to an exponential growth in the number of partial solutions,
as one such partial solution is created for each subset of transitions where the new solution
exceeds the old one. Usually, there are only few different partial solutions, i.e.\ we have just a
problem with space limitations but not really with computation time. Thus, we advise to modify the algorithm
such that the new partial solutions created from jump constraints are not introduced all at once but one at a time, each time
a former one has been processed. If a full solution is found, the unprocessed jump constraints
are deleted before they consume space and time, if no full solution exists our experience says
that jump constraints are scarce anyway.

\section{Diagnostic Information for Unreachability}\label{sec7}

If one of the optimizations~3 or~4 occurs, we know that an increment constraint $c$ added to our system of linear equations
did not have the desired effect. This means, we wanted to increase the number of tokens on some set of places but
were not able to do so, i.e. a $T$-invariant either was not firable or it had some side effect cancelling the usefulness
of the token increase. Such a thing could happen for example if we added a transition $t''$ to Fig.~\ref{f.tinv}
that consumes a token from each $s_2$ and $s_3$ and tried to fire it. The $T$-invariant $u+u'$ will produce a token on $s_2$, where we need
it, but takes away that token from $s_3$, cancelling its positive effect. Even if we fire this $T$-invariant more
than once, it is of no use.

Compare the partial solutions $ps=(\C,x,\sigma,r)$ and $ps'=(\C\cup\{c\},x+y,\sigma,r+y)$
resp.\ $ps'=(\C\cup\{c\},x+y,\sigma\tau,r)$ from optimizations~3 and~4. It is obvious that
adding the constraint $c$ had not the desired effect, i.e. $c$ failed. We were not able to
obtain enough tokens on some set of places $S_i$ (that was the reason for introducing $c$) to
enable some transitions $T_i$ (or $X_i$, see section~\ref{sec4}). If we memorize $S_i$ and 
$T_i$/$X_i$ and the number of tokens $n$ missing on $S_i$ together with $c$, we can now give a 
(partial) reason for unreachability. We need at least $n$ more tokens on the set $S_i$
to fire one transition in $T_i$ or all transitions in $X_i$ after the firing sequence $\sigma$.
Indeed, $S_i$ produces a kind of deadlock: $S_i \cup T_i$ forms a component of the net $N$
that is strongly connected and the transitions of $T_i\cup X_i$ cannot fire after $\sigma$
since there are not enough tokens left in $S_i$. There may be other transitions though, that could put
tokens onto $S_i$. If our partial solutions could not be extended with them to obtain a
full solution, these transitions are either dead after firing $\sigma$ or they lead to a
system of state equation plus constraints that can not be fulfilled. In the latter case,
after firing such a transition the final marking becomes unreachable.

Even if we find a reason for unreachability, there may still be other paths that lead
to a solution. But if the algorithm terminates altogether (which we are not able to 
guarantee) and we have not found any solution, we can pick up our failed constraints
and present them as diagnostic information. Fig.\ref{f.counternet} in the next section shows a visualization
of such information in a small example Petri net. Hence, our tool cannot
only tell the user that a certain marking is unreachable but can also give hints as to where the
bottlenecks of the token distribution are that cannot be passed.

\section{Experimental Results}\label{sec8}

The algorithm presented here has been implemented in a tool named Sara \cite{sara}. We compare Sara
to LoLA~\cite{Wolf_2007_icatpn}, a low level analyzer searching the (reduced) state space of a Petri net. According to
independent reports, e.g.~\cite{talcottdill}, LoLA performs very well on reachability queries
and possibly is the fastest tool for standard low level Petri nets.
The following tests, real-world examples as well as academic constructions, were run on a 2.6GHz PC 
with 4GB RAM under Windows XP and Cygwin. While the CPU had four cores, only one was used for the tools.
Tests on a similar Linux system lead to comparable but slightly faster results.

\begin{iteMize}{$\bullet$}
\item 590 business processes with about 20 up to 300 actions each were tested for ``relaxed
soundness''. Relaxed soundness means that, for each transition $t$ of the net, it is possible to reach the final marking from the initial marking with a path that contains $t$. 
This problem can be solved using the methods presented above for reachability. The occurrence of $t$ can be asserted by an additional constraint $t > 0$ to the IP solver.
The processes were transformed into Petri nets and for each action a test was performed
to decide if it was possible to execute the action and reach the final state of the process afterwards.
Successful tests for all actions/transitions yield relaxed soundness.
Sara was able to decide relaxed soundness for all of the 590 nets together (510 were relaxed sound)
in 198 seconds, which makes
about a third of a second per net. One business process was especially hard and took 12278 calls
to lp\_solve and 24 seconds before a decision could be made. LoLA was unable to solve 17
of the problems (including the one mentioned above) and took 24 minutes for the remaining 573. 
\item Four Petri nets derived in the context of verifying parameterized boolean programs (and published on a web page \cite{wsts}) were presented to us to decide coverability.
Sara needed less than one time slice of the CPU per net and solved all instances correctly.
LoLA was not able to find the negative solution to one of the problems due to insufficient memory (here, tests were made with up to 32GB RAM),
the remaining three problems were immediately solved.
\item In 2003, H. Garavel~\cite{garavel03} proposed a challenge on the internet to check a Petri net derived from
a LOTOS specification for dead (i.e. never firable) transitions. The net consisted of 776 transitions
and 485 places, so 776 tests needed to be made. Of the few tools that succeeded, LoLA was the fastest
with about 10 minutes, but it was necessary to handle two of the transitions separately with a differently
configured version of LoLA. In our setting, seven years later, LoLA needed 41 seconds to obtain the
same result. Sara came to the same conclusions in 26 seconds. In most cases the first solution of
lp\_solve was sufficient, but for some transitions it could take up to 15 calls
to lp\_solve. Since none of the 776 transitions is dead,
Sara also delivered 776 firing sequences to enable the transitions, with an average length of
15 and a longest sequence of 28 transitions. In 2003 the best upper bound for the sequences
lengths was assumed to be 35, while LoLA found sequences of widely varying length (the longest having several thousand transition occurrences), though most
were shorter than 50 transitions.
\item We investigated five nets that represent biochemical reaction chains. The nets stem from the Pathway Logic Assistent \cite{talcottdill} where LoLA
is integrated for solving reachability problems. For some of the nets (which have several hundred places and transitions), LoLA was incapable of
verifying reachability. Sara was able to solve the problems. For three systems, Sara ran less than a second, the remaining two problems
could be solved in approximately half an hour. In one of these problem instances the given marking was unreachable.
\item Using specifically constructed nets with increasing arc weights (and token numbers) it was possible to outsmart
Sara -- the execution times rose exponentially with linearly increasing arc weights, the first five
times being 0.1, 3.3, 32, 180, and 699 seconds. LoLA, on the other hand,
decided reachability in less than 3 seconds (seemingly constant time) in these cases.
\end{iteMize}\smallskip

\noindent In addition to these experiments, Sara participated in the Model Checking Contest \cite{mcc} that was organized within the workshop on
Scalable and Usable Model Checking for Petri Nets and Other Models of Concurrency in June 2011 in Newcastle upon Tyne.
In this contest, Sara competed in reachability queries on place/transition Petri nets representing a flexible manufacturing system,
a KANBAN system, and a biochemical reaction chain. Experiments were done on increasing state spaces that were obtained by adding
additional tokens (up to several hundred) on certain places. Sara was able to solve most problem instances in less than a second, the remaining cases
were solved within a few seconds. Only two queries led to memory overflow. Sara outperformed all other participating tools including LoLA.

\begin{table}[tb]
\centering\small
\begin{tabular}{|l|r|r|c|c|c|c|c|}\hline
Net & Inst. & Sol? & Full & $\neg$1 & $\neg$2 & $\neg$3/4 & $\neg$5\\\hline
garavel & 776 & 15 & 26s (0.11) & 25s (0.11) & 26s (0.11) & 26s (0.11) & 26s (0.11)\\\hline
bad-bp & 142 & - & 24s (85) & 24s (85) & 24s (85) & 24s (85) & NR\\\hline
good-bp & 144 & 53 & 1.7s (0) & 1.7s (0) & 1.7s (0) & 1.7s (0) & 1.7s (0)\\\hline
test7 & 10 & 175 & 29s (13) & 990s (22) & NR & 49s (14) & 29s (13)\\\hline
test8-1 & 1 & 40 & 0.1s (13) & 0.35s (22) & 49s (13) & 0.2s (14) & 0.11s (13)\\\hline
test8-2 & 1 & 76 & 3.3s (21) & 24s (51) & NR & 11s (34) & 3.8s (21)\\\hline
test8-3 & 1 & 112 & 32s (27) & 390s (80) & NR & 175s (71) & 33s (27)\\\hline
test9 & 1 & - & 0.4s (53) & 22s (464) & NR & NR & 0.9s (65)\\\hline
\end{tabular}\vspace*{2mm}
\caption{\label{t.1}Results for shutting down one heuristic. Inst.\ is the number of problem
instances to be solved for the net, Sol? the average solution length or ``-'' if no solution
exists. Columns Full, $\neg$1, $\neg$2, $\neg$3/4, and $\neg$5 contain the result
with all optimizations, without stubborn sets, without
subtree cutting, without partial solution cutting, and without saving intermediate results 
(numbers are according to Sect.~\ref{sec5}).
Each entry shows the elapsed time and the number of necessary CEGAR steps (average), or
NR if no result could be obtained in less than a day}
\end{table}

We also checked our heuristics from Sect.~\ref{sec5} with some of the above nets by switching
the former off and comparing the results (see Table~\ref{t.1}). Our implementation needs both forms of constraints,
jump and increment, to guarantee that all solutions of the state equation can be visited.
Going through these solutions in a different order, e.g.\ the total order $\ord$, is difficult
and a comparison was not possible so far.

The nets tested fall in two categories. Garavel's net 
and the business processes are extensive nets with a low token count and without much
concurrency that could be tackled by partial order reduction. The heuristics have no effect
here, short runtimes result from finding a good solution to the state equation early on.
Only for the hardest of the business processes (bad-bp) memorizing intermediate results
to avoid checking the same partial solution over and over made sense -- without it we did not get
a result at all.

The other category are compact nets. In our test examples a high number of tokens is produced 
and then must be correctly distributed, before
the tokens can be removed again to produce the final marking.
With a high level of concurrency in the nets, partial order
reduction is extremely useful, the cutting off of already seen subtrees(2) even more than
the stubborn set method(1). In the last net (test9), the sought intermediate token distribution is
unreachable but the state equation has infinitely many solutions. Only by cutting off
infinite parts of the solution tree with the help of optimization~3 and~4 it becomes
possible to solve the problem at all. Without them, the number of outstanding CEGAR
steps reaches 1000 within less than a minute and continues to increase monotonically.
The algorithm slows down more and more then as the solutions to the state equation and thus
the potential firing sequences become larger.

As mentioned in the previous section, Sara can also provide diagnostic information
for unreachable problem instances as long as the state equation has a solution. This feature was tested
e.g.\ with the hardest of the 590 business processes from above, which provides such a negative
case for some of its 142 transitions. Using the diagnostic information, we were able to understand the
reason for unreachability thus validating the result computed by Sara.
Since we cannot present such a large net here, a condensed
version with the same important features is shown in Fig.~\ref{f.counternet}.

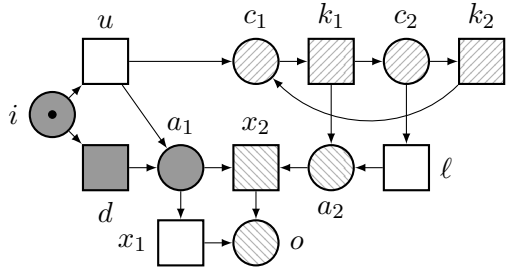
\begin{figure}[tb]
\centering
\begin{tikzpicture}[scale=0.5]
\node[gplace,label=left:$i$] (i) {};
\node[transition,label=above:$u$] (u) [above right of=i] {};
\node[gtransition,label=below:$d$] (d) [below right of=i] {};
\node[iplace,label=above:$c_1$] (c1) [right of=u, xshift=1cm] {};
\node[gplace,label=above:$a_1$] (a1) [right of=d] {};
\node[itransition,label=above:$k_1$] (k1) [right of=c1] {};
\node[iplace,label=above:$c_2$] (c2) [right of=k1] {};
\node[itransition,label=above:$k_2$] (k2) [right of=c2] {};
\node[otransition,label=above:$x_2$] (x2) [right of=a1] {};
\node[oplace,label=below:$a_2$] (a2) [right of=x2] {};
\node[transition,label=right:$\ell$] (l) [right of=a2] {};
\node[transition,label=left:$x_1$] (x1) [below of=a1] {};
\node[oplace,label=right:$o$] (o) [right of=x1] {};
\draw (i) node[token] {};
\draw[arrow] (i) to (u);
\draw[arrow] (i) to (d);
\draw[arrow] (u) to (c1);
\draw[arrow] (u) to (a1);
\draw[arrow] (d) to (a1);
\draw[arrow] (c1) to (k1);
\draw[arrow] (k1) to (c2);
\draw[arrow] (c2) to (k2);
\draw[arrow] (k2) to[out=225,in=315] (c1);
\draw[arrow] (k1) to (a2);
\draw[arrow] (c2) to (l);
\draw[arrow] (l) to (a2);
\draw[arrow] (a1) to (x2);
\draw[arrow] (a2) to (x2);
\draw[arrow] (a1) to (x1);
\draw[arrow] (x1) to (o);
\draw[arrow] (x2) to (o);
\end{tikzpicture}
\caption{\label{f.counternet}A condensed, flawed business process. One token should flow
	from the initial place $i$ to the output place $o$ with all other places empty finally.
	Non-white transitions
	appear in Sara's solution to the state equation, but only the dark gray one is fireable.
	Ascending stripes show the area with non-fireable transitions where additional tokens could not be generated}
\end{figure}

Sara provides a partitioning of the net showing where the relaxed soundness test (for any of the
transitions $k_1$, $k_2$, or $x_2$) fails, e.g.\ it is impossible to fire $x_2$ and afterwards reach the
final marking with exactly one token on place $o$ (other places being empty). The solution 
$d+k_1+k_2+x_2$ of the state equation
can neither be realized nor extended to a ``better'' solution.
The ascending pattern shows a region of the net (given by Sara)
where tokens are needed but cannot be generated without violating the state equation.
The descending pattern marks areas that are affected by the former ones, i.e. areas with
also non-fireable transitions. The gray transition $d$ is the only fireable transition occurring
in the solution. When analyzing the net we can see that the cycle $c_1-k_1-c_2-k_2$ indeed
constitutes a flaw for a business process: if the cycle gets marked and then emptied later,
at least two tokens must flow through $a_2$, one of which can never be removed.
Using $u$ instead of $d$ is therefore impossible, i.e.\ $d x_1$ is the only
firing sequence reaching the final marking.

\section{Towards Parallel Execution}

In essence, Sara solves and evaluates a large number of IP problems corresponding to partial solutions.
Different partial solutions are processed mostly independently. In addition, Fig.~\ref{f.solpath} suggests that
our search space is in fact a tree. It is thus natural to consider a parallelization scheme where new threads are
opened whenever there is more than one possibility to add a constraint to a partial solution. 

The network traffic for opening a new thread is rather low. The Petri net input as well as the problem instance can be 
loaded by each thread independently and ahead of the actual activation. Then, only the description of a partial
solution, especially a set of additional constraints, needs to be transmitted. Compared to the network traffic,
the internal work involves solving at least one IP problem including subsequent analysis. 

Hence, the crucial factor for feasibility of
this approach is the branching factor in the search space. With branching factor, we mean the number of new subproblems
that are introduced upon the analysis of a partial solution. If the branching factor is one, the subproblems appear sequentially
and parallelization does not make sense. Larger values suggest better parallelization results.

We measured the branching factors in several of the examples listed in the previous section. Among those instances that were
solved by Sara in less then a second, the maximum branching factor was two while the average branching factor was close to one.
In nontrivial examples with practical background (e.g.~the hard instances of biochemical reaction chains), the maximal branching factor was
two. The average branching factor ranged between 1.25 and 1.99, so there is enough room for feeding many independent threads.
Substantial speedup can be expected for processing the large number of IP instances (beyond 160.000).
In academic challenges, the maximum branching factor was four, with 2.6 being a typical average value.

We conclude, that the potential for parallel execution is excellent, especially in those problem instances where it is
most needed. 

\section{Conclusion}

We proposed a promising technique for reachability verification. It is based on the Petri net state equation that naturally
provides an overapproximation of the set of reachable states of  Petri net. Using the idea of counterexample
guided abstraction refinement, we were able to significantly improve the precision of the technique. Our approach has
several advantages compared to state space techniques:

(1) It is very efficient, for two reasons. First, we traverse the set of solutions of the state equation rather than the set of all
reachable states. That is, our search space is already substantially constrained. Second, we employ the very mature
technology of IP solving. We could validate the efficiency in a large number of challenging examples as well as in the
independent assessment made in the model checking contest in 2011.

(2) It tends to produce very small (not necessarily minimal) witness paths. This is due to the gradual introduction of
period vectors to a minimal solution of the state equation. Unlike explicit state space techniques, short witnesses come
without exponentially blowing up the search space. In this regard, the state equation as such resembles a symbolic
state representation. Indeed, one solution to the state equation represents up to exponentially many different firing
sequences.

(3) It has the tendency to terminate early on unreachable problem instances. In several cases, the initial state equation
may already assert unreachability. In contrast, state space techniques (whether explicit or symbolic)
cannot benefit from the on-the-fly paradigm if
the target state is unreachable. So they would produce all the reachable states modulo the applied state space reduction
techniques.

(4) It has a rather pleasant memory consumption. As IP solving is an NP complete problem, polynomial space is sufficient for
the core procedure. Only the management of open subproblems may require arbitrary space.

(5) It has an excellent potential for parallelization. Internal executions (IP solving) are nontrivial while network traffic (transmitting
a problem description) is rather lightweight.

(6) It produces some kind of diagnostic information for unreachable problem instances. This potential must be further explored.
In particular, usefulness of the diagnostics must be assessed in studies with independent users unaware of the
internal mechanisms of Sara.

On the negative side, we need to mention that our approach is incomplete. We are not able to show termination of our procedure
and a guaranteed termination may even contradict the EXPSPACE hardness of the reachability problem in general. On the
other hand, our experience with Sara so far is very encouraging.

Our approach applies the concept of counterexample guided abstraction refinement in a novel context: the
abstraction is not given as a transition system but as a linear-algebraic overapproximation of the reachable states.

The state equation as such has been used earlier for verification purposes, see for instance \cite{melzer00}.
In \cite{narrowing}, it is used as an initial way of narrowing the state space exploration but not refined
according to the CEGAR.

\bibliographystyle{plain}
\bibliography{reachability}

\end{document}